\documentclass[11pt]{article}
\usepackage{amsmath,amsfonts,amssymb,amsthm,mathtools}
\usepackage{bm}
\usepackage{mathrsfs}
\usepackage{algorithm}
\usepackage{algpseudocode}
\usepackage{array}
\usepackage{float}
\usepackage[font=footnotesize]{subfig}
\usepackage{textcomp}
\usepackage{graphicx}
\usepackage{booktabs}
\usepackage{enumitem}
\usepackage{url}
\usepackage{verbatim}
\usepackage{cite}
\usepackage{hyperref}
\hypersetup{
    colorlinks=true,
    linkcolor=blue,
    citecolor=blue,
    urlcolor=blue,
    pdfborder={0 0 0}   
}
\usepackage[all]{hypcap}
\usepackage[capitalise,nameinlink,noabbrev]{cleveref}

\crefname{algorithm}{Algorithm}{Algorithms}
\Crefname{algorithm}{Algorithm}{Algorithms}
\crefname{algocf}{Algorithm}{Algorithms}
\Crefname{algocf}{Algorithm}{Algorithms}
\crefname{figure}{Fig.}{Figs.}
\Crefname{figure}{Fig.}{Figs.}

\DeclareMathOperator{\rowspanop}{rowspan}
\newcommand{\rspan}[1]{\rowspanop\!\left(#1\right)}
\DeclareMathOperator{\im}{im}
\DeclareMathOperator{\Span}{span}

\DeclareMathOperator{\rank}{rank}
\DeclareMathOperator{\supp}{supp}
\DeclareMathOperator{\range}{range}

\DeclareMathOperator{\spark}{spark}
\DeclareMathOperator{\diag}{diag}

\newcommand{\norm}[1]{\left\lVert#1\right\rVert}
\newcommand{\R}{\mathbb{R}}
\newcommand{\Prb}{\mathbb{P}}

\newcommand{\Id}{\mathbb{I}}
\newcommand{\ba}{\bm{a}}
\newcommand{\bx}{\bm{x}}
\newcommand{\by}{\bm{y}}
\newcommand{\be}{\bm{e}}
\newcommand{\eS}{\bm{e}^\star}
\newcommand{\bv}{\bm{v}}

\newcommand{\bz}{\bm{z}}

\newcommand{\bu}{\bm{u}}
\newcommand{\bw}{\bm{w}}

\newcommand{\bP}{\bm{P}}

\newcommand{\bA}{\bm{A}}
\newcommand{\bF}{\bm{F}}
\newcommand{\bU}{\bm{U}}
\newcommand{\bB}{\bm{B}}

\newcommand{\bC}{\bm{C}}
\newcommand{\bZ}{\bm{Z}}

\newcommand{\bQ}{\bm{Q}}

\newcommand{\sG}{\mathscr{G}}

\newcommand{\cZ}{\mathcal{Z}}

\newcommand{\cR}{\mathcal{R}}

\newcommand{\cV}{\mathcal{V}}

\newcommand{\tr}{\top}

\newcommand{\xS}{\bm{x}^\star}

\newcommand{\Output}{\item[\textbf{Output:}] }

\newcommand{\sE}{\mathscr{E}}

\theoremstyle{definition}
\newtheorem{definition}{Definition}
\newtheorem{example}{Example}
\theoremstyle{plain}
\newtheorem{theorem}{Theorem}
\newtheorem{proposition}{Proposition}

\newtheorem{corollary}{Corollary}
\newtheorem{claim}{Claim}
\theoremstyle{remark}
\newtheorem*{remark*}{Remark}
\newtheorem{remark}{Remark}

\begin{document}

\title{\texorpdfstring{%
\makebox[\textwidth][c]{%
\parbox{1.2\textwidth}{\centering
Robustness to Sparse Adversarial Corruption in\\
Arbitrary Linear Measurements: Beyond Exact Recovery}}%
}{Robustness to Sparse Adversarial Corruption in Arbitrary Linear Measurements: Beyond Exact Recovery}}

\author{%
\makebox[\textwidth][c]{%
\begin{minipage}{1.2\textwidth}
\centering
\small Vishal Halder\textsuperscript{\(\dagger\),*}, Alexandre Reiffers-Masson\textsuperscript{\(\dagger\)}, Abdeldjalil A\"issa-El-Bey\textsuperscript{\(\dagger\)}, and Gugan Thoppe\textsuperscript{\(\ddagger\)}\\[0.5em]
\small \textsuperscript{\(\dagger\)}IMT Atlantique, Lab-STICC, CNRS UMR 6285, Brest, France\\
\small \textsuperscript{\(\ddagger\)}Department of Computer Science and Automation, Indian Institute of Science, Bengaluru, India\\[0.5em]
\scriptsize \{vishal.halder, alexandre.reiffers-masson, abdeldjalil.aissaelbey\}\mbox{@}imt-atlantique.fr; gthoppe\mbox{@}iisc.ac.in
\end{minipage}%
}%
}

\date{\today}

\maketitle
\begingroup
\renewcommand{\thefootnote}{*}
\footnotetext[1]{Corresponding author.}
\endgroup

\begin{abstract}

Recovery from linear measurements under sparse adversarial corruption is typically formulated as an exact-recovery problem: one seeks structural conditions on $\bA$ 
(e.g., restricted isometry property) 
guaranteeing unique recovery of $\xS$ from $\by = \bA\xS + \be$ with $\|\be\|_0 \leq q$. 
However, 
these guarantees provide no guidance once exact recovery fails. This limitation obscures simple robustness phenomena---for instance, repeated rows in $\bA$ can preserve nontrivial information about $\xS$ under sparse corruption.

In this paper, we study what information about $\xS$ can be \emph{uniformly} recovered from $\by = \bA\xS + \be$ for arbitrary $\bA\in\mathbb{R}^{m\times n}$ and \emph{any} $q$-sparse $\be$. We show that the robust information is precisely $\xS + \ker(\bU)$, where $\bU$ is the orthogonal projection onto the intersection of rowspaces of all submatrices of $\bA$ obtained by deleting $2q$ rows. This clarifies how the row structure of $\bA$ governs whether a $q$-sparse corruption allows exact, partial, or only trivial recovery. We further prove every $\bx$ minimizing $\|\by - \bA \bx\|_0$ belongs to $\xS + \ker(\bU)$, yielding a constructive approach to recover this set. For i.i.d.\ Gaussian matrices, we establish a sharp phase transition between exact and trivial recovery. We sketch two applications: robust network tomography and signal reconstruction from oversampled DCT.

\end{abstract}


\section{Introduction}
\label{sec:intro}

In many signal-processing systems, a small subset of measurements may be compromised and behave adversarially. 
Examples arise in network monitoring~\cite{amin2009safe}, where compromised routers may falsify traffic data; power systems~\cite{liu2011false,xu2013sparse}, where attacked meters may inject false readings; and distributed sensing~\cite{ren2020secure}, where hijacked sensors may report adversarial measurements. 
A central objective in robust signal processing has therefore been to identify conditions under which the true signal can still be reconstructed from grossly corrupted measurements. 
Such guarantees are often expressed through structural assumptions on the measurement operator, including restricted isometry~\cite{candes2005decoding}, nullspace conditions~\cite{fawzi2014secure}, subset strong convexity/smoothness~\cite{bhatia2015robust}, and spark-type conditions~\cite{lee2018redundant}. 
These conditions provide exact-recovery guarantees when they hold, but they can be restrictive in applications and are widely reported to be NP-hard to verify \cite{tillmann2013computational, bandeira2013certifying, weed2017approximately}. 

More fundamentally, this exact-recovery viewpoint leaves open an important intermediate regime: a measurement system may fail standard recovery conditions and yet still encode certain projections, aggregates, or linear functionals of the signal in a way that is immune to sparse corruption. 
We study this question for the linear model
\begin{equation}\label{eq:linear-model}
    \by = \bA\xS + \be,
\end{equation}
where $\bA\in\mathbb{R}^{m\times n}$ is known, $\xS\in\mathbb{R}^n$ is unknown, and $\be\in\mathbb{R}^m$ is an unknown corruption vector satisfying $\|\be\|_0\le q$.

\begin{example}\label{ex:basic}
   Let 
\[
\bA=\begin{bmatrix}
1 & 1 & 1 & 0 & 0\\
0 & 0 & 0 & 1 & 1\\
1 & 1 & 1 & 1 & 1\\
1 & 1 & 1 & 0 & 0\\
0 & 0 & 0 & 1 & 1
\end{bmatrix}.
\]
Notice that here $\bA$ is not full column rank, so exact recovery of $\xS$ is impossible even in the absence of corruption. Now suppose that any one coordinate of the measurement vector is arbitrarily corrupted, so that $\by = \bA\xS +\be,$ where $\be \in \R^5$ is unknown with $\|\be\|_0\leq q= 1.$ Can we
identify any information about $\xS$ from $\by$ that remains robust despite any corruption $\be$? The answer is yes: the quantities
$s=\frac{1}{3}\sum_{i=1}^3 x_i^\star$ and
$t=\frac{1}{2}\sum_{i=4}^5 x_i^\star$
remain robust. Indeed, since
\begin{equation}
\bA\xS = 
\begin{bmatrix}
    s \quad t \quad s+t \quad s \quad t
\end{bmatrix}^\tr,
\end{equation}
if any one of the coordinates of $\bA\xS$ is corrupted, we can nevertheless determine the true values of $s$ and $t$ by a majority vote. In fact, as we shall see, these quantities constitute all robust information about $\xS$ in this example (see the discussion after \Cref{thm:tight-orth-proj}).
\end{example}

\Cref{ex:basic} demonstrates a critical gap between exact-recovery guarantees and robust partial recovery.
The question addressed in this paper is:

\medskip
\textit{Given an arbitrary measurement matrix $\bA\in\R^{m\times n}$ and an integer $q$, what is the maximum information about an unknown $\xS \in \R^n$ that is guaranteed to be robust in the measurement $\by = \bA \xS + \be$ despite any sparse corruption $\be\in\R^m$ with $\|\be\|_0 \le q$?}
\medskip

In \Cref{thm:tight-orth-proj}, we show that the answer is determined by a canonical subspace of $\R^n$ associated with $\bA$ and $q$. For each $T\subseteq[m]$, let $\bA_T$ denote the submatrix of $\bA$ with rows indexed by $T$, and define
\begin{equation}\label{defn:robust-subspace}
    \cR := \bigcap_{\substack{T\subseteq[m],\\ |T|=m-2q}} \rspan{\bA_T}.
\end{equation}
In \eqref{defn:robust-subspace}, considering submatrices with $2q$ row deletions is natural. Indeed, suppose two signals $\bx_1$ and $\bx_2$ can explain the same measurement $\by$ under $q$-sparse corruptions. Then there exist $\be_1$, $\be_2$ with $\|\be_1\|_0, \|\be_2\|_0 \le q$ such that $\by = \bA\bx_1 + \be_1 = \bA\bx_2 + \be_2$. This implies $\|\bA(\bx_1 - \bx_2)\|_0 \le 2q$, so the signal ambiguity $\bx_1 - \bx_2$ lies in the nullspace of some submatrix of $\bA$ obtained by deleting at most $2q$ rows. This notion is formalized in \Cref{subsec:ambiguity,subsec:the-robust-subspace}.
We call $\cR$ the \emph{robust subspace}. 

Specifically, in \Cref{thm:tight-orth-proj}, we show that the smallest set containing $\xS$ that we can robustly recover is
\begin{equation}\label{eq:intro-smallest-solution-set}
        \xS+\ker(\bU),
\end{equation}
where $\bU$ denotes the orthogonal projector onto $\cR$. 
 Equivalently, the orthogonal projection $\bU\xS$ is a canonical representative of the robust information: it is the minimum-norm element of the set $\xS+\ker(\bU)$. 

Our result \eqref{eq:intro-smallest-solution-set} has a natural form, analogous to familiar solution sets in simpler settings:
in the absence of corruption ($\by = \bA \xS$), the 
solution set
is $\xS + \ker(\bA)$, 
since any $\xS + \bv$ with $\bv \in \ker(\bA)$ produces the same measurements;
if the corruption support $I \subseteq [m]$ with $|I| \leq q$ were known, then the solution set would be $\xS + \ker(\bA_{[m] \setminus I})$, obtained by discarding the corrupted measurements. \eqref{eq:intro-smallest-solution-set} establishes the analogue of these results when the $q$-sparse corruption is arbitrary and unknown.

\subsection{Our Contributions}\label{sec:innovations}

This paper introduces the robust subspace as the object that characterizes the information surviving sparse adversarial corruption. Our main contributions are as follows.

\begin{enumerate}
\item In \Cref{thm:tight-orth-proj}, we prove that the maximal uniformly robust information about $\xS$ is $\bU\xS$, where $\bU$ is the orthogonal projector onto $\cR.$
Equivalently, the smallest robust solution set containing $\xS$ is $\xS+\ker(\bU)$.

\item We show in \Cref{thm:l0} that $\ell_0$-decoding leads to this minimal set, by proving that 
    $\xS + \ker(\bU) = \hat{\bx}_0 + \ker(\bU),$
for any $\hat{\bx}_0 \in \arg \min_x \|\by - \bA\bx\|_0$.

\item We give an exact deterministic algorithm for computing $\bU$ (\Cref{alg:projector}) and prove that computing the robust subspace is NP-hard for arbitrary matrices (\Cref{cor:np-hard}).

\item We identify structured regimes in which the robust subspace is explicit or efficiently computable, including i.i.d.\ Gaussian matrices (\Cref{sec:gauusian-random-matrices}), and oversampled orthonormal-transform measurements (\Cref{subsec:DCT}).

\item \emph{Gaussian Random Matrices}\\
In \Cref{thm:gaussian} we show a sharp phase transition behavior in the robust subspace for matrices with i.i.d. standard normal entries: above the corruption sparsity fraction threshold of $q \geq (m - n)/2$, the robust subspace $\cR$ collapses from $\R^n$ to $\{0\}$. Thus, one either has (theoretical) full recovery or no (nontrivial) recovery---there is no intermediate regime.

\item We sketch two applications of our results in \Cref{sec:numerical}.
\begin{enumerate}[label=(\alph*)]

\item \emph{Robust network tomography}: 
For path--link incidence matrices, $\bU$ identifies the link quantities and link aggregates uniformly recoverable under $q$ corrupted path measurements.

\item \emph{Oversampled DCT atoms}: 
For row-sampled orthonormal transforms, $\cR$ is spanned by atoms appearing with multiplicity greater than $2q$, so $\bU$ is computable in $O(m)$ time.

\end{enumerate}
\end{enumerate}
We also clarify the relation between our work and robust subspace recovery, showing that despite superficial similarities, the two problems address fundamentally different recovery objectives. And finally, we show that the robust function and solution-set ideas extend beyond sparse adversarial corruption by deriving analogous assumption-free recovery results for compressed sensing. 

\section{Formalizing Robustness}
\label{sec:preliminary-robustness}
To precisely characterize what information about an input signal $\xS$ remains robust in a measurement $\by,$ we need a rigorous notion of robustness and robust information. First, note that under adversarial corruption, the same measurement may admit multiple signal--corruption explanations, as illustrated below.
\begin{example}\label{ex:ambiguity}
\[
\underbrace{
\begin{bmatrix}
1 & 0\\
1 & 0\\
1 & 0\\
0 & 1\\
0 & 1
\end{bmatrix}
\begin{bmatrix}
1\\
2
\end{bmatrix}
+
\begin{bmatrix}
0\\
0\\
0\\
-1\\
0
\end{bmatrix}}_{\bA\bx_1+\be_1}
\;=\;
\underbrace{
\begin{bmatrix}
1 & 0\\
1 & 0\\
1 & 0\\
0 & 1\\
0 & 1
\end{bmatrix}
\begin{bmatrix}
1\\
1
\end{bmatrix}
+
\begin{bmatrix}
0\\
0\\
0\\
0\\
1
\end{bmatrix}}_{\bA\bx_2+\be_2}.
\]
Here, the inputs $\bx_1 $ and $\bx_2$ are indistinguishable under the linear measurement model \eqref{eq:linear-model} with $q=1$.
\end{example}
In this section, we first formalize a robust quantity as the value of a function that is invariant across all signal--corruption pairs consistent with the same measurement, and call such functions robust functions. We then characterize the set of all signal vectors in $\R^n$ that can cause such ambiguity by explaining the same measurement (the ambiguity set), and then establish that a necessary and sufficient characterization of robust functions is their invariance over this set.

\subsection{Robust Functions and Solution Sets}
A quantity (output of a function) derived from the input signal is said to be robust only if it has the same value for all signals that are consistent with the same corrupted measurement. 
 \begin{definition}
\label{def:robust}
For $\bA\in\R^{m\times n}$ and integer $q<m/2$, a function $\sG:\R^n\to\cZ$ (arbitrary codomain) is said to be $(\bA, q)$-\emph{robust} if for every $\bx_1, \bx_2 \in \R^n,$ and $q$-sparse $\be_1, \be_2 \in \R^m$,
\begin{equation}
    \bA\bx_1 + \be_1 = \bA\bx_2 + \be_2 \; \implies \;  \sG(\bx_1) = \sG(\bx_2).
\end{equation}
\end{definition}

The information about $\xS$ that is captured by a robust function $\sG$ is characterized by its solution set, defined below.
\begin{definition}\label{def:rob-sol-set}
    Let $\bA \in \R^{m \times n}$ and $q < m/2$ be an integer. Then, for any $(\bA, q)$-robust function $\sG$ and any $\xS \in \R^n$, the \emph{$\sG$-robust solution set containing $\xS$} is  
    \begin{equation}
    X(\sG,\xS) := \{\bx \in \R^n : \sG(\bx) = \sG(\xS)\}.
    \end{equation}
    Clearly, $X(\sG, \xS) = \sG^{-1} (\sG(\xS)),$ the preimage of $\sG(\xS)$ under $\sG$. 
\end{definition}

Any pair of input signals $\bx_1$ and $\bx_2$ that are consistent with the same measurement $\by$, even if the signals themselves differ, would lie in the same robust solution set of any given robust function. That is, signals $\bx_1$ and $\bx_2$ are indistinguishable to any robust function $\sG$, in addition to being indistinguishable to the measurement $\by$. 

\subsection{The Ambiguity Set}\label{subsec:ambiguity}
To formalize the fundamental indistinguishability under the measurement model $\by = \bA\xS + \be$, we introduce the ambiguity set. The ambiguity set consists of all $\bv=\bx_1-\bx_2$  for which there exists $\be_1$ and $\be_2$ with $\|\be_1\|_0,\|\be_2\|_0\leq q$ such that
\begin{equation}
    \bA\bx_1+\be_1=\bA\bx_2+\be_2.
\end{equation} 
\begin{definition}
\label{def:ambiguity-set}
For $\bA\in\R^{m\times n}$ and integer $q<m/2$, define the \emph{ambiguity set} $S_{q}^{\bA} := \{\bv\in\R^n : \|\bA\bv\|_0 \le 2q\}$.
\end{definition}
\begin{remark}\label{rem:ambiguityset-2q-remark}
    The bound $2q$ in the definition of the ambiguity set arises because indistinguishability between two signals $\bx_1$ and $\bx_2$ means that there exist corruptions $\be_1$ and $\be_2$ with $\|\be_1\|_0,\|\be_2\|_0\leq q$ such that $\bA\bx_1+\be_1=\bA\bx_2+\be_2$. Their difference $\be_1 - \be_2$ therefore has cardinality of support at most $2q$, which implies $\|\bA(\bx_1 - \bx_2)\|_0 \le 2q$.
\end{remark}

In the following result, we show that a fundamental characteristic of robust functions is their invariance under addition of elements of the ambiguity set. 

\begin{proposition}
\label{thm:general-recoverability}
Let $\bA\in\R^{m\times n}$, integer $q<m/2$, and $S_{q}^{\bA} := \{\bv\in\R^n:\|\bA\bv\|_0\le 2q\}$. A function $\sG:\R^n\to\cZ$ is $(\bA, q)$-robust iff $\sG(\bx+\bv)=\sG(\bx)$ for all $\bx\in\R^n$ and $\bv\in S_{q}^{\bA}$.
\end{proposition}
\begin{proof}
($\Rightarrow$) For $\bv\in S_{q}^{\bA}$, let $T=\supp(\bA\bv)$ with $|T|\le 2q$. Partition $T$ as $T = T_1 \cup T_2$. Next, define $\be'$ and $\be$ by $(\be')_j = (\bA\bv)_j$ if $j\in T_1$, 0 otherwise, and $(\be)_j = -(\bA\bv)_j$ if $j\in T_2$, 0 otherwise. Then, for any $\bx \in \R^n$ and $\bv \in S_{q}^{\bA},$ we have $\bA(\bx + \bv)+\be = \bA \bx +\be'$ with $\|\be\|_0,\|\be'\|_0\le q$, so robustness implies $\sG(\bx+\bv)=\sG(\bx)$.  

($\Leftarrow$) Let $\bx$ and $\bx'$ be such that $\bA\bx + \be = \bA\bx'+\be'$ with $\|\be\|_0,\|\be'\|_0\le q.$ Then, $\bA(\bx' - \bx) = \be - \be'$ has support $\le 2q,$ which implies $\bx' - \bx \in S_{q}^{\bA}.$ Therefore, by hypothesis, $\sG(\bx) = \sG(\bx + \bx' - \bx) = \sG(\bx'),$ as desired. 
\end{proof}
Notice that $\xS + S_q^{\bA}$ lies in the robust solution set containing $\xS$ for every robust function $\sG$. 
In \Cref{thm:tight-orth-proj}, we show that the smallest such robust solution set is $\xS + \ker(\bU),$ where $\bU$ is the orthogonal projection matrix onto the robust subspace.

\subsection{The Robust Subspace}
\label{subsec:the-robust-subspace}
For the measurement model $\by = \bA \xS + \be$ with $\|\be\|_0 \leq q,$ the subspace 
\(
    \cR = \bigcap_{\substack{T \subset [m], \\ |T| = m - 2q}} \rspan{\bA_T},
\)
is the robust subspace. The reason $2q$ row deletions arise here is that $\cR$ is orthogonal to every vector in the ambiguity set, i.e.,
\begin{equation}\label{robsupace-ambiguityset-orthogonality}
    \bv \perp \cR, \quad \forall\bv \in S_{q}^{\bA}.
\end{equation}

Indeed, fix $\bv\in S_q^{\bA}$ and let
\(
Q := \supp(\bA\bv),
\)
so that $|Q|\le 2q$. By definition of $Q$, we have $(\bA\bv)_i=0$ for every $i\notin Q$, and hence there are at least $m-|Q|\ge m-2q$ indices on which $\bA\bv$ vanishes. Choose any set $T\subset[m]$ with $|T|=m-2q$ such that $(\bA\bv)_i=0$ for all $i\in T$; equivalently,
\begin{equation}
    \bA_T\bv=0.
\end{equation}

Thus $\bv\in\ker(\bA_T)$. By the fundamental theorem of linear algebra,
\(
\ker(\bA_T)=\rspan{\bA_T}^\perp.
\)
Since $\cR \subseteq \rspan{\bA_T}$ for every such $T$, it follows that $\bv\perp \cR$, i.e., $\bv\in\cR^\perp$.

In our result \Cref{thm:tight-orth-proj}, we improve on \eqref{robsupace-ambiguityset-orthogonality} and show that, in fact, 
\(
    \bv \perp \cR,  \forall\bv \in \Span(S_{q}^{\bA}),
\)
i.e., \eqref{eq:R^perp}. Using this, we obtain the inclusion-wise minimal robust solution set as $\xS + \Span(S^{\bA}_q).$

\section{Main results}
\label{sec:main}

We now state and prove our main results: \Cref{thm:tight-orth-proj,thm:l0}. \Cref{thm:tight-orth-proj} shows that, for any $\bA$ and $q$, there exists a linear function $ \sG(\bx)=\bU\bx$ such that  (1) $\bU$ is an orthogonal projection matrix, (2) $\bU$ is $(\bA,q)$-robust, and (3) $\bU$ yields an inclusion-wise minimal robust solution set: $X(\bU,\xS)\subseteq X(\sG,\xS)$ for every $(\bA,q)$-robust $\sG$ and every $\xS \in \R^n$. Seeking an inclusion-wise minimal robust solution set is equivalent to finding the set that provides the maximum robust information about $\xS$ from a given measurement $\by = \bA\xS + \be$.

\begin{theorem}\label{thm:tight-orth-proj}
Consider a matrix $\bA \in \R^{m\times n}$ and an integer $q < m/2$. Also, let $\cR = \bigcap_{\substack{T \subseteq [m], \\ |T| = m - 2q}}\rspan{\bA_T}$ and $\bU$ the orthogonal projector onto $\cR.$ That is, suppose $\bU$ is the unique symmetric idempotent matrix with $\range(\bU) = \cR$ and $\ker(\bU) = \cR^{\perp}.$ 
Then, the $\bx \rightarrow \bU \bx$ map is $(\bA, q)$-robust. Moreover, for any $(\bA, q)$-robust $\sG:\R^n\to\cZ$ and $\xS\in\R^n$, 
\begin{equation}
    \{\bx\in\R^n : \sG(\bx) = \sG(\xS)\} \supseteq \xS + \ker(\bU),
\end{equation}
with equality for $\sG(\bx) = \bU \bx $.
\end{theorem}
\begin{proof}
By the fundamental theorem of linear algebra \cite{strang1993fundamental}, $\rspan{\bA_T}=(\ker(\bA_T))^\perp$, and since $(\cap \cV_i)^\perp=\sum \cV_i^\perp$ (see \cite[Ex.~5, Section12]{Halmos1958}), we obtain  $
\cR^\perp=\sum_{\substack{T \subseteq [m], \\ |T| = m - 2q}}\ker(\bA_T),$ 
where the sum of vector spaces is the Minkowski sum. 
For $\bv\in S_q^{\bA}$ with $Q=\supp(\bA\bv)$ and  $|Q|\le 2q$, we have $\bv\in\ker(\bA_{Q^c})$, so 
\begin{equation}
S_q^{\bA}=\bigcup_{\substack{Q \subseteq [m], \\ |Q| \le 2q}} \ker(\bA_{Q^c}), \quad 
\Span(S_q^{\bA})=\sum_{\substack{Q \subseteq [m], \\ |Q| \le 2q}}  \ker(\bA_{Q^c})
\end{equation}
(since the span of a union of subspaces is their sum, see \cite[Thm.~3, Section11]{Halmos1958}).  
Writing $T=Q^c$ implies $|T|\ge m-2q$, and using $\ker(\bA_{T_2})\subseteq\ker(\bA_{T_1})$ for $T_1\subseteq T_2$, it follows that 
\begin{equation}
    \label{eq:R^perp}\Span(S_q^{\bA})=\sum_{\substack{T \subseteq [m], \\ |T| = m - 2q}}\ker(\bA_T)=\cR^\perp.
\end{equation}
Thus, \begin{equation}\label{eq:ker=span}
    \ker(\bU)=\Span(S_q^{\bA}),
\end{equation}
which implies that $\bU\bv=0$ whenever $\bv\in S_q^{\bA}.$ \Cref{thm:general-recoverability} now shows that the map $\sG^\star(\bx) = \bU \bx$ is $(\bA, q)$-robust.  
To prove the final statement, 
let $\xS \in \R^n$ and $\sG$ be $(\bA, q)$-robust. 
Any $\bw \in \Span(S_q^{\bA})$ can be written as $\bw = \sum_{j=1}^r c_j \bv_j$ with $\bv_j \in S_q^{\bA}$ and $c_j \in \R$.  
Since $c_j \bv_j \in S_q^{\bA}$, we apply \Cref{thm:general-recoverability} iteratively: first $\sG(\xS + c_1 \bv_1) = \sG(\xS)$, then $\sG(\xS + c_1 \bv_1 + c_2 \bv_2) = \sG(\xS + c_1 \bv_1) = \sG(\xS)$, and so on until $\sG(\xS + \sum_{j=1}^r c_j \bv_j) = \sG(\xS)$.  
Hence, $\xS + \Span(S_q^{\bA}) \subseteq \{\bx \in \R^n : \sG(\bx) = \sG(\xS)\}$.
\end{proof}

This theorem shows that the robustness of a matrix $\bA$ against unknown $q$-sparse adversaries is characterized exactly by the deletion of $2q$ rows of $\bA$. In particular, the theorem identifies the subspace  
\begin{equation}
    \cR = \bigcap_{\substack{T \subseteq [m], \\ |T| = m - 2q}} \rspan{\bA_T}
\end{equation}
as the object governing robustness. We refer to this intersection as the \emph{robust subspace}, and to its orthogonal projector $\bU$ as the \emph{robust orthogonal projection matrix}. The affine subspace $\xS + \ker(\bU)$ being the smallest robust solution set leads to the fact that $\bU\xS$ is the uniquely determined maximal robust partial information about $\xS$ with minimum $\ell_2$ norm.

Continuing the discussion on \Cref{ex:basic}, observe that the subspace spanned by rows of the observation matrix $\bA$, $\Span(\bu, \bv),$ with $\bu=\begin{bmatrix}1\, 1\, 1\, 0\, 0\end{bmatrix}^\tr$ and $ \bv=\begin{bmatrix}0\,0\,0\,1\,1\end{bmatrix}^\tr,$ remains invariant to any $2q=2$ row deletions. Further, notice that the projection of $\xS$ onto $\Span(\bu, \bv)$ gives $(s,t)^\tr,$ with $s=\frac{\bu^\tr \xS}{\|\bu\|^2}$ and $t = \frac{\bv^\tr \xS}{\|\bv\|^2}.$ Therefore, based on \Cref{thm:tight-orth-proj}, we conclude that $\bU\xS = (s,t)^\tr$, with
\[
s=\frac{\bu^\tr \xS}{\|\bu\|^2}=\frac{1}{3}\sum_{i=1}^3 x_i^\star, \quad
t=\frac{\bv^\tr \xS}{\|\bv\|^2}=\frac{1}{2}\sum_{i=4}^5 x_i^\star,
\]
is the most one can hope to recover robustly.

\begin{remark} Our results in \Cref{thm:tight-orth-proj} are universal in the sense that they fully characterize what can be deterministically recovered without any assumptions on 
$\bA$. They also generalize prior results on exact signal recovery under structured measurement matrices. Indeed, \Cref{thm:tight-orth-proj} recovers the necessary and sufficient exact-recovery condition stated in \cite[Proposition 2]{fawzi2011secure} as a special case when $\bA$ satisfies conditions equivalent to the robust orthogonal projection being identity, i.e., $\bU = \Id_n$.
\end{remark}
\begin{corollary}\label{cor:L0-fawzi}
    Consider a matrix $\bA \in \R^{m\times n}$ and an integer $q < m/2$. Also, let $\cR = \bigcap_{\substack{T \subseteq [m], \\ |T| = m - 2q}}\rspan{\bA_T}$ and $\bU$ the orthogonal projector onto $\cR.$ Then, 
    \begin{equation}
    \ker(\bU) = \{0\} \, \Leftrightarrow \, |\supp(\bA\bz)| > 2q \; \forall \bz \in \R^n\setminus\{0\}.
    \end{equation}
\end{corollary}
\begin{proof}
From \Cref{thm:tight-orth-proj}, $\ker(\bU) = \Span(S_q^{\bA})$. Thus, $ \ker(\bU) = \{0\}\, \Leftrightarrow \, \Span(S_q^{\bA}) = \{0\}\, \Leftrightarrow \, S_q^{\bA} = \{0\}$. But $S_q^{\bA} = \{0\}$ iff there is no nonzero $\bz$ with $\|\bA\bz\|_0 \le 2q$, i.e., iff $|\supp(\bA\bz)| > 2q$ for all $\bz \in \R^n \setminus \{0\}$.
\end{proof}

Given \Cref{thm:tight-orth-proj}, which establishes that the inclusion-wise minimal robust solution set that is robust to $q$ adversarial corruptions is $\xS + \ker(\bU)$, the central question now becomes how to extract this set. Clearly, this set is equivalent to $\{\bx \in \R^n: \bU \bx = \bU \xS\}.$ Our next result shows how we can extract this set, knowing $\bU$ and any $\ell_0$ minimizer. 

\begin{theorem}\label{thm:l0}
Let $\by = \bA \xS + \be$ with $\|\be\|_0 \le q$, and let $\bU$ be the orthogonal projection matrix onto $\bigcap_{\substack{T \subseteq [m], \\ |T| = m - 2q}}\rspan{\bA_T}$.
Then, every $\hat\bx \in \arg\min_{\bx\in \R^n} \|\by - \bA \bx\|_0$ satisfies
$   \hat\bx  + \ker(\bU) = \xS + \ker(\bU).$
\end{theorem}
\begin{proof}
Let $\hat\be=\by-\bA\hat\bx$. Since $\hat\bx$ minimizes $\|\by-\bA\bx\|_0$, we have $\|\hat\be\|_0 \le \|\be\|_0 \le q$.  
Thus $\supp(\hat\be-\be)\subseteq\supp(\hat\be)\cup\supp(\be)$, giving $\|\hat\be-\be\|_0 \le 2q$.  
But $\hat\be-\be=(\by-\bA\hat\bx)-(\by-\bA\xS)=\bA(\xS-\hat\bx)$, so $\xS-\hat\bx\in S_q^{\bA}$.  
By \eqref{eq:ker=span}, $S_q^{\bA}\subseteq\ker(\bU)$, hence $\bU\hat\bx=\bU\xS$. Now clearly, $\hat\bx - \xS \in \ker(\bU),$ i.e., $\hat\bx  + \ker(\bU) = \xS + \ker(\bU)$.
\end{proof}

\section{Computation and Hardness}

In \Cref{thm:l0,} we see that the robust orthogonal projection matrix $\bU$ and $\ell_0$-decoding together yield a constructive procedure to recover all robust information about $\xS$. This section is devoted to the problem of computing $\bU$, and equivalently, the robust subspace $\cR.$

\subsection{Computing the Robust Orthogonal Projection}
\label{sec:alg}

Here, we introduce a computation scheme (\Cref{alg:projector}) that takes as input any arbitrary measurement matrix $\bA$ and an upper bound $q$ on the sparsity of the adversarial corruption $\be$, and returns the robust orthogonal projection matrix.

\begin{algorithm}[t]
\caption{Computing the robust orthogonal projector}
\label{alg:projector}
\begin{algorithmic}[1]
\Require $\bA\in\R^{m\times n}$, integer $q<m/2$
\Output Orthogonal projector $\bU$ onto the robust subspace $\bigcap_{\substack{T \subseteq [m], \\ |T| = m - 2q}} \rspan{\bA_T}$
\State $\mathscr T \gets \{T\subseteq[m]:|T|=m-2q\}$, $\bC\gets 0_{n\times n}$
\For{$T\in\mathscr T$}
    \State $\bA_T \gets$ submatrix of $\bA$ with rows in $T$
    \State $\bB_T \gets$ orthonormal basis of $\ker(\bA_T)$ via $\mathrm{svd}(\bA_T)$
    \State $\bC \gets \bC + \bB_T\bB_T^\top$
\EndFor
\State $[\bQ,\bm{\Lambda}]\gets\mathrm{eig}(\bC)$ 
\State $\bZ \gets$ columns of $\bQ$ with eigenvalue $\lambda=0$
\State \Return $\bU \gets \bZ\bZ^\top$ 
\end{algorithmic}
\end{algorithm}

\begin{claim}
\Cref{alg:projector} outputs the orthogonal projector onto the robust subspace
$\cR = \bigcap_{\substack{T \subseteq [m], \\ |T| = m - 2q}}\rspan{\bA_T}.$
\end{claim}
\begin{proof}
Let $\cR = \bigcap_{\substack{T \subseteq [m], \\ |T| = m - 2q}}\rspan{\bA_T}.$ For each $T\subseteq [m]$ with $|T| = m - 2q$, let $\bP_T := \bB_T \bB_T^\top$ be the orthogonal projector onto $\ker(\bA_T)$, so $\im(\bP_T) = \ker(\bA_T)$. Then $\im(\bC) = \im(\sum_T \bP_T) = \sum_T \im(\bP_T) = \sum_T \ker(\bA_T)$, since in finite dimensions the image of a sum equals the sum of images. Hence $\ker(\bC) = (\im(\bC))^\perp = (\sum_T \ker(\bA_T))^\perp$. By the orthogonal-complement identity $(\sum_i V_i)^\perp = \cap_i V_i^\perp$ \cite[Ex.~5, Section12]{Halmos1958}, and by the fundamental theorem of linear algebra $\rspan{\bA_T} = (\ker(\bA_T))^\perp$ 
\cite{strang1993fundamental}
, we have $\ker(\bC) = \cap_T (\ker(\bA_T))^\perp = \cap_T \rspan{\bA_T} = \cR$. Therefore the zero-eigenvectors of $\bC$ span $\cR$, and $\bU = \bZ \bZ^\top$ is the orthogonal projector onto $\cR$.
\end{proof}

While not polynomial-time, \Cref{alg:projector} is conceptually valuable as an exact deterministic procedure to compute the robust orthogonal projector for any $(\bA, q)$, achieving the set inclusion-wise minimality guarantee established in \Cref{thm:tight-orth-proj}. Although it is a combinatorial scheme, in the next section, we prove that we can do no better for the case of arbitrary matrices, since computing the robust subspace in this general case is NP-Hard.

\subsection{Hardness}
\label{sec:hardness}

Here, we show that the problem of computing the robust subspace $\cR$ (and the associated robust orthogonal projection matrix $\bU$) for any given arbitrary matrix $\bA \in \R^{m \times n}$ and integer $q<m/2$ is NP-hard in general. 
Note that although NP-hardness holds true for arbitrary $\bA,$ we can still have fast algorithms to compute the robust subspace for structured observation matrices. For instance, we shall see in \Cref{subsec:DCT} that this computation is possible in $O(m)$ time for an $\bA$ with orthogonal rows, whereas for the canonical class of i.i.d. Gaussian random matrices, we see in \Cref{sec:gauusian-random-matrices} that the robust subspace can be determined in $O(1)$ time.

We prove NP-hardness through the decision question of whether the robust subspace is rank deficient for any \(\bA\in\mathbb{Q}^{m\times n}\) and \(q\), i.e.,
\begin{equation}\label{decision:RSRD}
\text{Is } \dim\Bigl(\bigcap_{\substack{T\subseteq[m]\\ |T|=m-2q}}\rspan{\bA_T}\Bigr)<n\ ?
\end{equation}
\eqref{decision:RSRD} is equivalent to $\mathrm{MinULR}_0^=$, the NP-complete problem of deciding whether, given \(\bA\in\mathbb{Q}^{m\times n}\) and \(k\), there exists \(\bv\neq 0\) such that \(\norm{\bA\bv}_0\le k\)~\cite{tillmann2019computing}. Indeed, setting \(k=2q\) and using \eqref{eq:R^perp},
$\dim(\cR)<n
\iff \cR^\perp\neq\{0\} 
\iff \Span(S_q^{\bA})\neq\{0\} 
\iff \exists\,\bv\neq0:\ \norm{\bA\bv}_0\le 2q .$

Thus deciding \eqref{decision:RSRD} is NP-hard. And, a polynomial-time algorithm for computing \(\cR\) for arbitrary \(\bA\) would yield a polynomial-time algorithm for \eqref{decision:RSRD}.
Hence, the following result is immediate.
\begin{theorem}\label{cor:np-hard}
    For arbitrary $\bA\in \R^{m \times n}$ and integer $q<m/2$, computing the linear subspace $\cR=\bigcap_{\substack{T \subseteq [m], \\ |T| = m - 2q}}\rspan{\bA_T}$ is NP-hard.
\end{theorem}

\section{Gaussian Random Matrices}\label{sec:gauusian-random-matrices}
In this section, we argue that with probability $1$, the robust subspace $\cR$ for Gaussian measurement matrices with independent and identically distributed entries shows a sharp phase transition phenomenon.
The key insights from our result \Cref{thm:gaussian} are as follows.

If the sparsity $q$ of the corruption vector $\be$ obeys $q\leq(m-n)/2$, then each row submatrix $\bA_T$ with $|T|=m-2q\geq n$ is a tall random matrix with i.i.d. normal entries. Thus, their rowspaces are expressive enough to be the full space $\R^n,$ leading to the robust subspace $\cR = \R^n.$

Whereas in the case where $q>(m-n)/2$, the wide submatrices $\bA_T$ with $|T|=m-2q<n$ are "disjoint enough" that only the $0$ vector remains common in their row spaces, showing collapse of the robust subspace to $\{0\}.$

Moreover, \Cref{thm:gaussian} exemplifies that the computation of the robust subspace for structured matrices is possible in polynomial time. For a Gaussian random matrix, knowing only its size and the corruption sparsity $q$, we have the robust subspace in $O(1)$ time. Another such example is given in \Cref{subsec:DCT}.

\begin{theorem}\label{thm:gaussian}
    Let $\bA \in \R^{m \times n} = (a_{ij})_{1\leq i \leq m, 1 \leq j \leq n}$ be a random matrix with i.i.d. standard normal entries, i.e.,  $a_{ij} \sim \mathcal{N}(0,1)$. Fix any integer $1 \leq q < m/2.$ Define the robust subspace as in \Cref{thm:tight-orth-proj}:
\(
\mathcal R=\allowbreak\bigcap_{\substack{T\subset[m]\\|T|=m-2q}}\allowbreak\rspan{A_T},
\)
and let $\bU$ be the orthogonal projector onto $\cR.$
Then, almost surely, the following hold true.
\begin{equation}\label{thm:gaussian-statement-1}
    \mathcal{R} = \R^n \quad \text{if and only if} \quad q \le (m-n)/2,
\end{equation}
\begin{equation}\label{thm:gaussian-statement-2}
    \mathcal{R} = \{0\} \quad \text{if and only if} \quad q > (m-n)/2.
\end{equation}
Equivalently, $\bU = \Id_n$ precisely when $q \le (m-n)/2$, and $\bU = 0$ precisely when $q > (m-n)/2$.
\end{theorem}
\begin{proof}
We prove the biconditionals by establishing only the "if" (sufficient) directions:
\begin{enumerate}
    \item \label{proof:gaussian-statement-1}
    If \(q\le (m-n)/2\), then
\(
\Prb\big(\cR=\R^n \big)=1,
\)
hence \(\bU=\Id_n\) almost surely.
    \item \label{proof:gaussian-statement-2}
    If \(q> (m-n)/2\), then
\(
\Prb\big( \cR=\{0\} \big)=1,
\)
hence \(\bU=0\) almost surely.
\end{enumerate}
Then "only if" (necessary) directions then follow. Indeed, let $\cR = \R^n,$ and suppose (for contradiction) that \(q> (m-n)/2\). But then by Statement \ref{proof:gaussian-statement-2} above, we would have $\cR = \{0\},$ leading to a contradiction. Similarly, let $\cR = \{0\},$ and suppose (for contradiction) that \(q \leq (m-n)/2\). Then, by Statement \ref{proof:gaussian-statement-1} above, we would have $\cR = \R^n,$ leading to a contradiction.

We now proceed to prove Statements \ref{proof:gaussian-statement-1} and \ref{proof:gaussian-statement-2}. \\
Denote \(k := m-2q\).

\begin{enumerate}
\item \label{proof:random-part-1} Suppose \(q \le (m-n)/2\), so \(k \ge n\).

Fix \(T \subseteq [m]\) with \(|T| = k\). Choose any \(T' \subseteq T\) with \(|T'| = n\), so \(
\rspan{\bA_T} \supseteq \rspan{\bA_{T'}}\). Now, the submatrix \(\bA_{T'} \in \R^{n \times n}\) is a square matrix with i.i.d. \(\mathcal{N}(0,1)\) entries. 

We proceed by using the following standard result (see, for instance, \cite[Section 1]{rudelson2008invertibility}, \cite[Eq.~(7.1)]{rudelson2009smallest}, \cite[Section 2.3]{vershynin2014invertibility}, \cite[Proof of Cor. 1.8]{candes2005decoding}, \cite[Proof of Thm. 8.2]{rudelson2013recent}): square random matrices with absolutely continuous entries are full rank almost surely.
Therefore, $\Prb(\det(\bA_{T'}) = 0) = 0,$ and hence 
\(
 \rspan{A_{T'}} = \R^n 
\) a.s.
Since  \(\rspan{\bA_T} \supseteq \rspan{\bA_{T'}}\), \(\rspan{\bA_T} = \R^n\) a.s.\ for every such fixed \(T\). Therefore,
\begin{align}
\, \Pr(\cR \ne \R^n) 
&= \Pr\bigl(\exists\, T\subset[m],\, |T|=k : \rspan{\bA_T} \ne \R^n\bigr) \\
&\le \sum_{\substack{T\subset[m]\\|T|=k}} \Pr(\rspan{\bA_T} \ne \R^n) = 0,
\end{align}
so \(\Prb(\cR = \R^n) = 1\).

\item Suppose \(q > (m-n)/2\), so \(k < n\). 

The proof strategy here is to show that there exists a row $\ba_r^\tr$ of $\bA,$ such that any vector from the robust subspace $\bv \in \cR$ can be written as a scalar multiple $\alpha \ba_r$. Then, linear independence of any $k+1\leq n$ rows of $\bA$  (proved in Part \ref{proof:random-part-1}), along with the fact that $\bv \in \rspan{\bA_T}$ for any $|T|=k$ (by definition of $\cR$), establishes that $\alpha=0,$ and thus $\bv = 0$. This is elucidated below.

Define a $k+1$ row index set as $I:=\{1, \dots, k+1\} \subset [m]$ without loss of generality.
Fix \(\bv \in \cR\). 
Remember that \(\bv \in \rspan{\bA_{I\setminus \{j\}}} = \Span\left(\{\ba_i\}_{i \in I\setminus \{j\}}\right)\) for each \(j \in \{1, \dots, k+1\}\), by definition of $\cR$ (since $|I\setminus \{j\}|=k)$.
We now show that \(\bv = \alpha \ba_{k+1}\) for some scalar \(\alpha \in \R\), by setting $l=k+1$ in the following identity:
\begin{equation}\label{eq:random-part-2-identity}
    \bv \in \Span\left(\{\ba_i\}_{l\leq i\leq k+1}\right)\quad \forall l \in \{1, \dots, k+1\}.
\end{equation}
We prove \eqref{eq:random-part-2-identity} by induction.
The case \(l=1\) is immediate, since for any $T\subset I$ with $|T| = k$, \(\rspan{\bA_T} \subseteq \Span\left(\{\ba_i\}_{1\leq i\leq k+1}\right)\), and so $\bv \in \cR \subseteq \Span\left(\{\ba_i\}_{1\leq i\leq k+1}\right).$
Assume \eqref{eq:random-part-2-identity} holds for some \(l < k\), i.e., there exist scalars $\alpha_i$ such that 
\begin{equation} \label{eq:rand-proof-part2-1}
    \bv = \sum_{i=l}^{k+1} \alpha_i \ba_i.
\end{equation} 
Since 
    $\bv \in \rspan{\bA_{I\setminus \{(l)\}}} = \Span\left(\ba_1, \dots, \ba_{l-1}, \ba_{l+1}, \dots, \ba_{k+1}\right),$
there exist scalars \(\beta_i\) such that
\begin{equation} \label{eq:rand-proof-part2-2}
    \bv = \sum_{i=1}^{l-1} \beta_i \ba_i + \sum_{i=l+1}^{k+1} \beta_i \ba_i.
\end{equation}
Equating \eqref{eq:rand-proof-part2-1} and \eqref{eq:rand-proof-part2-2}, and using the almost sure linear independence of \(\{\ba_1, \dots, \ba_{k+1}\}\) from Part \ref{proof:random-part-1}, we have \(\beta_i = 0\) for \(i \le l-1\), \(\alpha_{l} = 0\), and \(\alpha_i = \beta_i\) for \(i \ge l+1\). Thus, $\bv = \sum_{i=l+1}^{k+1} \alpha_i \ba_i \in \Span\left(\{\ba_i\}_{l+1\leq i\leq k}\right),$
completing the induction.
Next, since \(\bv \in \Span\left(\ba_1, \dots, \ba_k\right)\) by definition of $\cR$, the almost sure linear independence of \(\{\ba_1, \dots, \ba_{k+1}\}\) (from Part \ref{proof:random-part-1}) forces \(\alpha = 0\), i.e., \(\bv = 0\). Hence, \(\Prb(\cR = \{0\}) = 1\).
\end{enumerate}
\end{proof}

\begin{remark}\label{remark:finite-dim-result}
\Cref{thm:gaussian} is a finite-dimensional result. It holds simultaneously for every finite integer pair $(m,n)$.
Once \(m\) and \(n\) are fixed, the entire argument in the proof is non-asymptotic. Any infinite-dimensional analogue or high-dimensional asymptotic analysis would require different tools, e.g., concentration of measure in the large-\(n\) limit \cite{ledoux2001concentration}.
\end{remark}

\begin{remark}
The first assertion, equation \eqref{thm:gaussian-statement-1}, of \Cref{thm:gaussian} is essentially a restatement of Corollary 1.8 in \cite{candes2005decoding} in terms of the concept of the robust subspace introduced in this paper. Indeed, from \Cref{thm:tight-orth-proj}, we know that (theoretical) full recovery of \emph{any} \(\xS\in\R^n\) (i.e., unique recovery in the presence of any \(q\)-sparse corruption) is possible if and only if \(\ker(\bU)=\{0\}\), or equivalently \(\bU=\Id_n\). 
For i.i.d.\ Gaussian \(\bA\), \Cref{thm:gaussian-statement-1} says that this holds almost surely precisely when 
\(q\le(m-n)/2\), recovering the information-theoretic threshold 
\(\rho=q/m\le(1-n/m)/2\) stated in \cite[Corollary 1.8]{candes2005decoding}.
\end{remark}

\begin{figure}[H]
\centering

\subfloat[]{
\begin{minipage}{0.49\textwidth}
  \centering
  \includegraphics[width=\linewidth]{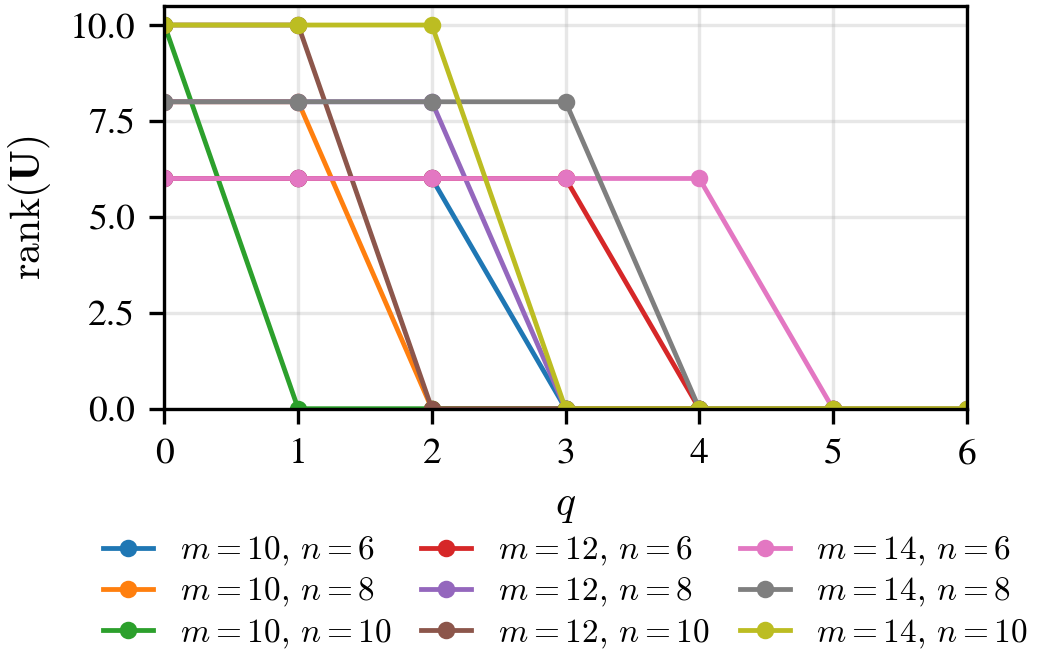}
  \label{fig:rank-vs-q}
\end{minipage}
}
\subfloat[]{
\begin{minipage}{0.48\textwidth}
  \centering
  \includegraphics[width=\linewidth]{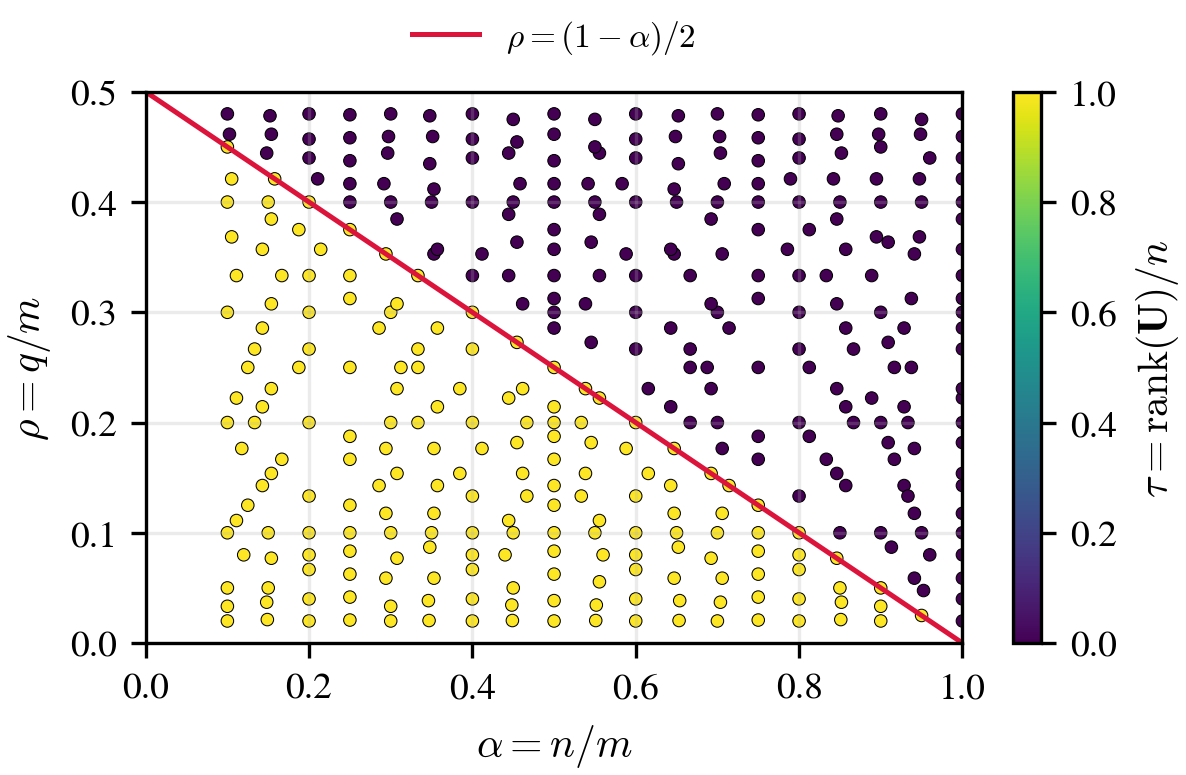}
  \label{fig:alpha-rho-phase}
\end{minipage}
}

\caption{Gaussian experiments. 
(a) $\rank(\bU)$ as a function of $q$ for nine Gaussian matrices with different $(m,n)$. 
(b) Normalized robust rank ($\rank(\bU)/n$) over corruption fraction $\rho=q/m$ and width:height $\alpha=n/m$.}
\label{fig:gaussian}
\end{figure}

We illustrate \Cref{thm:gaussian} empirically by conducting two experiments (\Cref{fig:gaussian}).

\textit{First experiment.} In \Cref{fig:rank-vs-q}, we generate $9$ different realizations of Gaussian matrices with i.i.d. standard normal entries with distinct $(m,n)$, and plot the rank of the robust orthogonal projection matrices computed using \Cref{alg:projector} for different values of $q$ for each $(m,n)$ pair. The collapse of the robust subspace from $\R^n$ to $\{0\}$ is evident, and it obeys the sharp phase transition of $(m-n)/2$.

\textit{Second experiment.} In \Cref{fig:alpha-rho-phase}, we visualize the sharp phase transition by plotting the signal-space normalized rank, $\rank(\bU)/n$ of the robust subspace for different values of the corruption fraction $\rho:=q/m$ and the width:height ratio of the matrices $\alpha:=n/m$. We generate realizations by sweeping through $m\in [3,50], n\in [1,50],$ and compute the robust subspace by varying $q\in [1,24]$ while obeying $q<m/2$, and using \Cref{alg:projector}. In \Cref{fig:alpha-rho-phase}, we report $\rank(U)/n$ for 389 $(\rho,\alpha)$ points. \Cref{fig:alpha-rho-phase} validates \Cref{thm:gaussian} empirically, and supports \Cref{remark:finite-dim-result}.

\section{Stylized Applications}
\label{sec:numerical}
We sketch a few application scenarios where our results may be of interest. 

\subsection{Robust Network Tomography} 
Consider now a problem in network tomography (cf. \cite{ma2013efficient}), where link level measurements (e.g., delay or packet loss) are to be inferred from path measurements. Consider the network shown in \Cref{fig:nt-sub-a}. The colored curves in the figure represent distinct observed measurement paths, where each path measurement is the sum of the link measurements on the path. Suppose that $q$ path measurements may be adversarially corrupted. The network manager would like to determine which link measurements are robust, and whether there exist linear combinations of links that are robust even when individual links are not. 
Our theory is particularly well suited to answer such questions. We first construct the standard $0$-$1$ path-link matrix for this network, then run \Cref{alg:projector} on it with $q$ as a parameter to obtain the robust orthogonal projection matrix $\bU$.

For $q=1$, the rank of $\bU$ is $3$, and a visualization of its basis vectors on the network is shown in \Cref{fig:nt-sub-c}. We find that the link measurements $x_3$ and $x_5$ are individually robust, and separately, that only the sum $x_4 + x_6 +x_7$ is robust. Under adversarial corruption, there is no recovery method that can recover the exact measurement on the remaining links.
Note that our theory can detect robust linear combinations of links even when they are disconnected (see basis vector $\bU_{. \, 4}$ in \Cref{fig:nt-sub-c}, corresponding to $x_4 + x_6 +x_7$).

Furthermore, the diagonal elements $\bU_{jj}$ of the robust projection matrix can be defined as \emph{robustness scores}, with $1$ meaning link $j$ is robust, $0$ meaning it is neither robust nor part of any robust linear combination, and a score in $(0,1)$ meaning link $j$ participates in a robust linear combination (e.g. $\bU_{. \, 4}$ in \Cref{fig:nt-sub-c}). For this network, the robustness scores are visualized in \Cref{fig:nt-sub-d} for the case $q=1.$ For $q>1$, the robust subspace collapses to the empty set and no link measurements are robust.

\begin{figure}[H]
\centering

\subfloat[]{%
\begin{minipage}[t]{0.42\textwidth}
  \vspace{0pt}
  \centering
  \includegraphics[width=\linewidth,height=0.7\textheight,keepaspectratio]{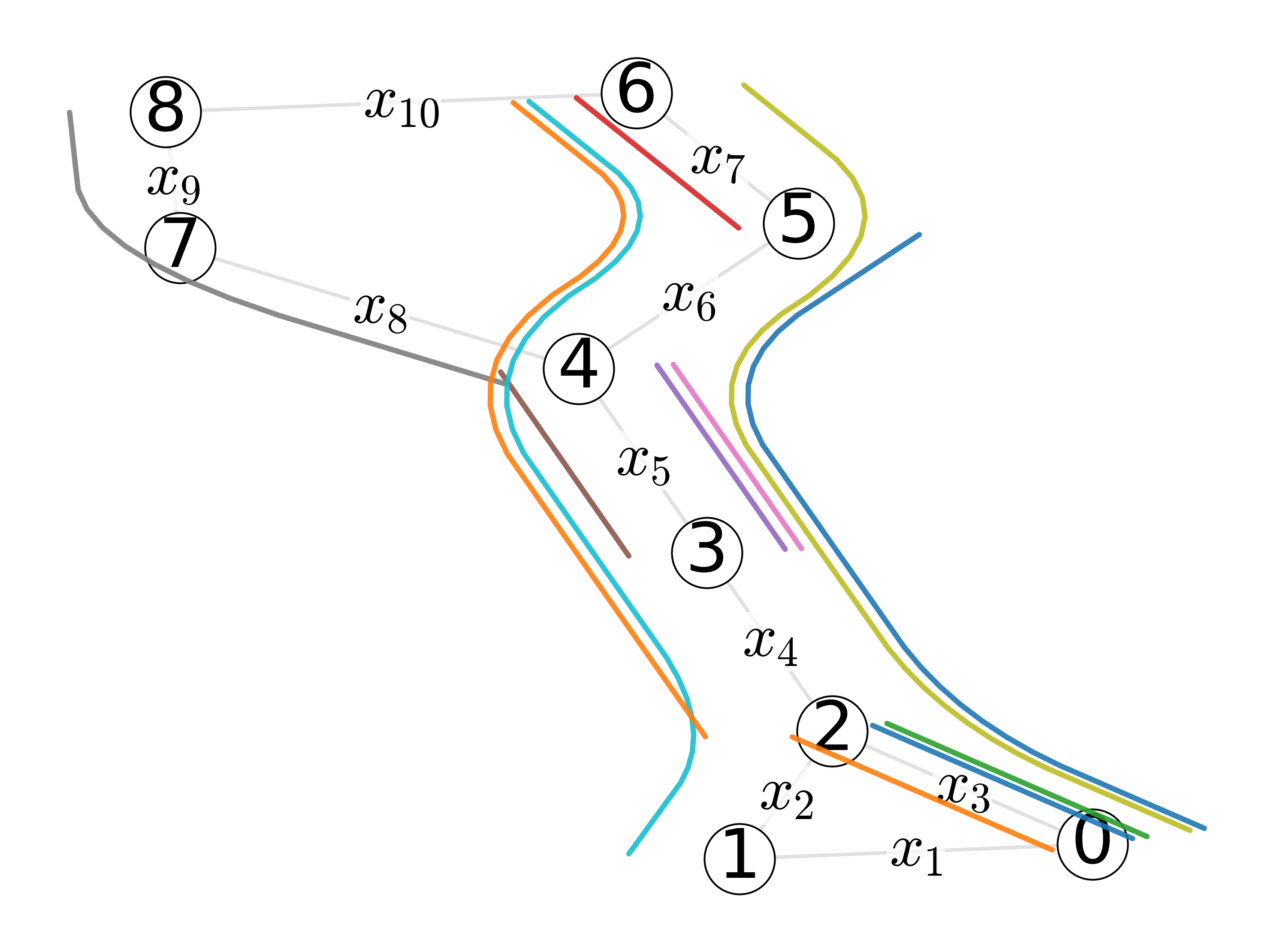}
  \label{fig:nt-sub-a}
\end{minipage}
}
\hfill
\subfloat[]{%
\begin{minipage}[t]{0.55\textwidth}
  \vspace{25pt}
  \centering
  \begingroup
  \setlength{\arraycolsep}{0.9pt}
  \resizebox{\linewidth}{!}{%
  \(
  \bU\bx
  \,=\,
    \begin{bmatrix}
      0 & 0 & 0 & 0 & 0 & 0 & 0 & 0 & 0 & 0 \\
      0 & 0 & 0 & 0 & 0 & 0 & 0 & 0 & 0 & 0 \\
      0 & 0 & 1 & 0 & 0 & 0 & 0 & 0 & 0 & 0 \\
      0 & 0 & 0 & \tfrac13 & 0 & \tfrac13 & \tfrac13 & 0 & 0 & 0 \\
      0 & 0 & 0 & 0 & 1 & 0 & 0 & 0 & 0 & 0 \\
      0 & 0 & 0 & \tfrac13 & 0 & \tfrac13 & \tfrac13 & 0 & 0 & 0 \\
      0 & 0 & 0 & \tfrac13 & 0 & \tfrac13 & \tfrac13 & 0 & 0 & 0 \\
      0 & 0 & 0 & 0 & 0 & 0 & 0 & 0 & 0 & 0 \\
      0 & 0 & 0 & 0 & 0 & 0 & 0 & 0 & 0 & 0 \\
      0 & 0 & 0 & 0 & 0 & 0 & 0 & 0 & 0 & 0
    \end{bmatrix}
    \begin{bmatrix}
      x_{1}\\ x_{2}\\ x_{3}\\ x_{4}\\ x_{5}\\ x_{6}\\ x_{7}\\ x_{8}\\ x_{9}\\ x_{10}
    \end{bmatrix}
    \,=\,
    \begin{bmatrix}
      0\\
      0\\
      x_{3}\\
      \tfrac13(x_{4}+x_{6}+x_{7})\\
      x_{5}\\
      \tfrac13(x_{4}+x_{6}+x_{7})\\
      \tfrac13(x_{4}+x_{6}+x_{7})\\
      0\\
      0\\
      0
    \end{bmatrix}
  \)
  }
  \endgroup
  \label{fig:nt-sub-b}
\end{minipage}
}

\par\medskip

\subfloat[]{%
\begin{minipage}[t]{\textwidth}
  \centering
  \includegraphics[width=\linewidth,height=0.2\textheight,keepaspectratio]{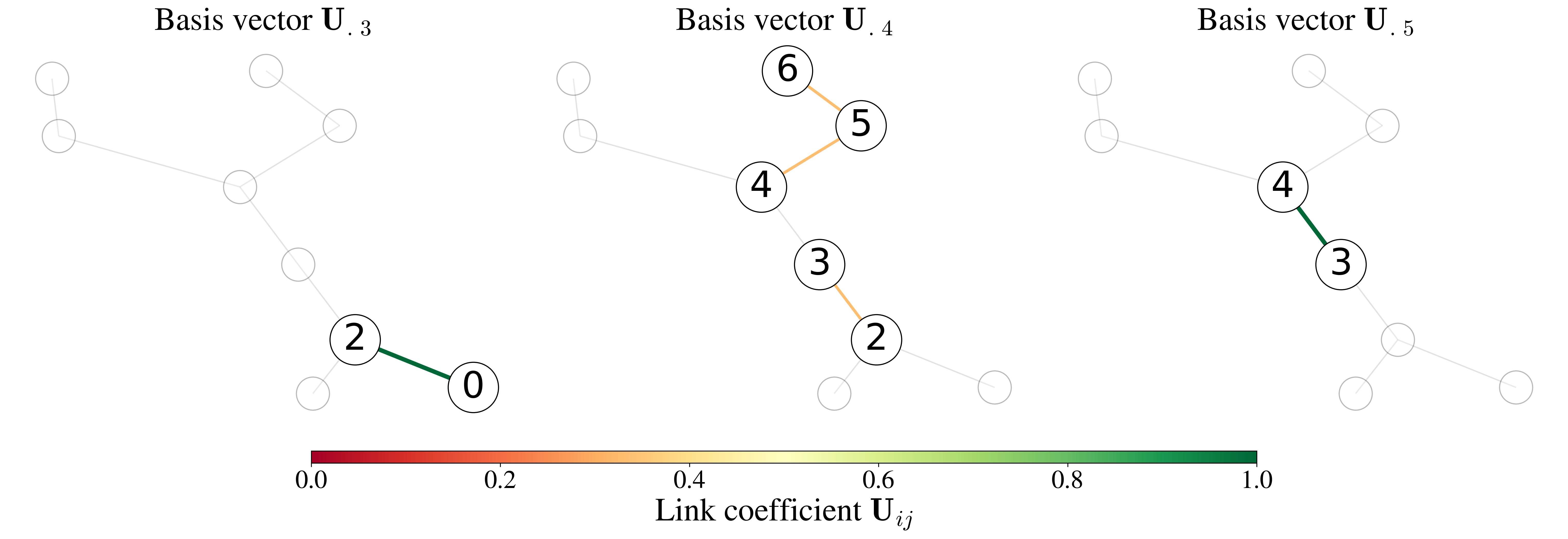}
  \label{fig:nt-sub-c}
\end{minipage}
}

\par\medskip

\subfloat[]{%
\begin{minipage}[t]{\textwidth}
  \centering
  \includegraphics[width=\linewidth,height=0.2\textheight,keepaspectratio]{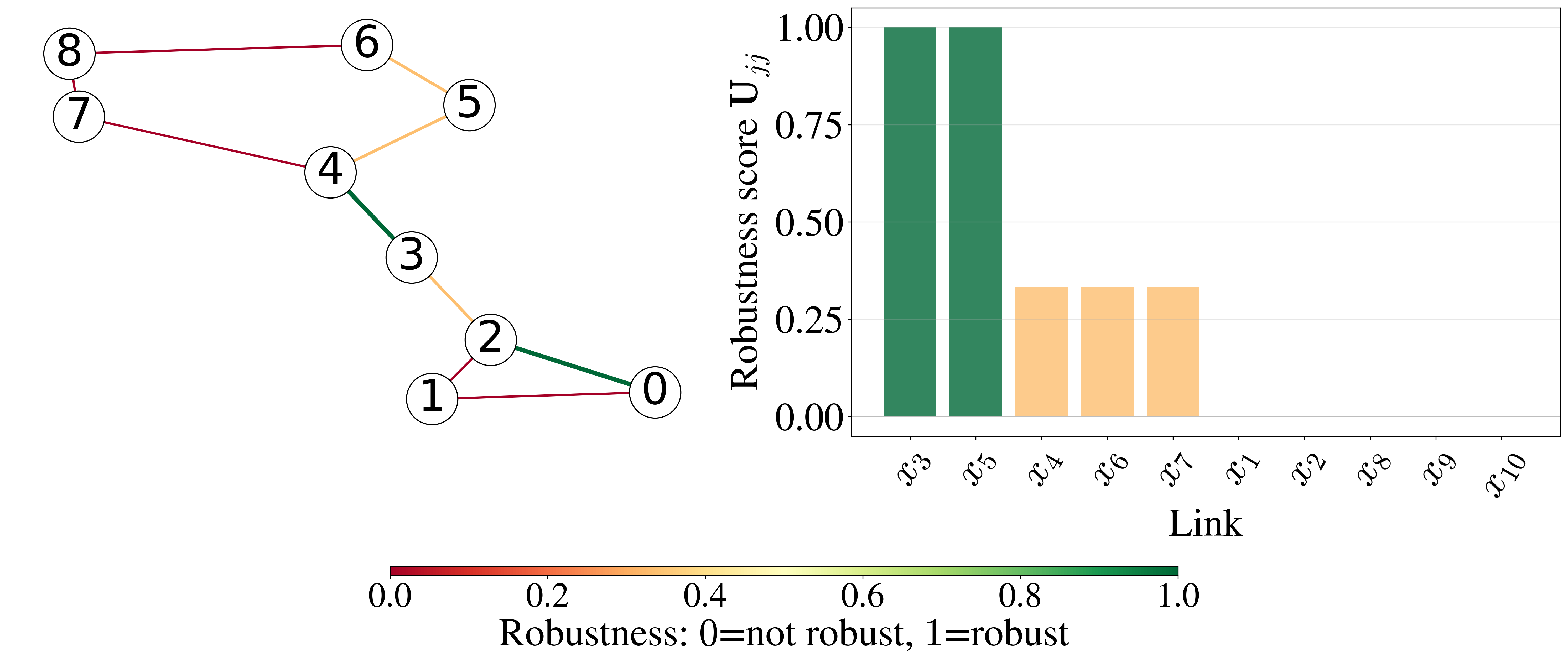}
  \label{fig:nt-sub-d}
\end{minipage}
}

\caption{Robust network tomography for $q=1$ corrupted path measurement.
(a) The network and observed paths.
(b) The robust projection matrix $\bU$ acting on a link-measurement vector $\bx$.
(c) Basis vectors of the robust subspace.
(d) Per-link robustness scores induced by $\diag(\bU)$.}
\label{fig:nt-2x2}
\end{figure}

\subsection{Oversampled DCT Measurements}
\label{subsec:DCT}
\begin{figure}[H]
\centering
\captionsetup[subfloat]{captionskip=-4pt}

\subfloat[]{%
\begin{minipage}[t]{0.47\textwidth}
  \centering
  \includegraphics[width=\linewidth,height=0.23\textheight,keepaspectratio]{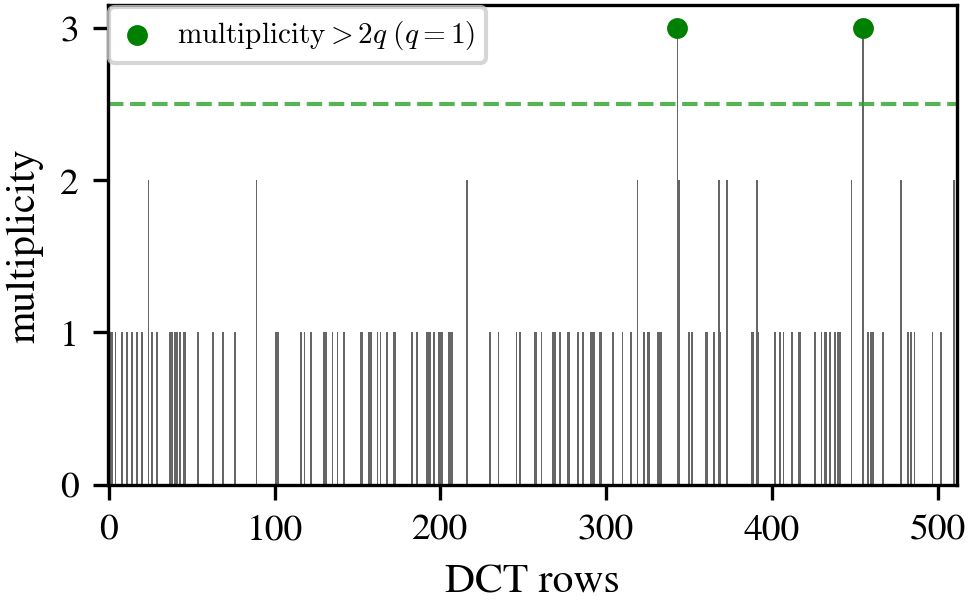}
  \label{fig:dct-sub-a}
\end{minipage}
}
\hfill
\subfloat[]{%
\begin{minipage}[t]{0.47\textwidth}
  \centering
  \includegraphics[width=\linewidth,height=0.23\textheight,keepaspectratio]{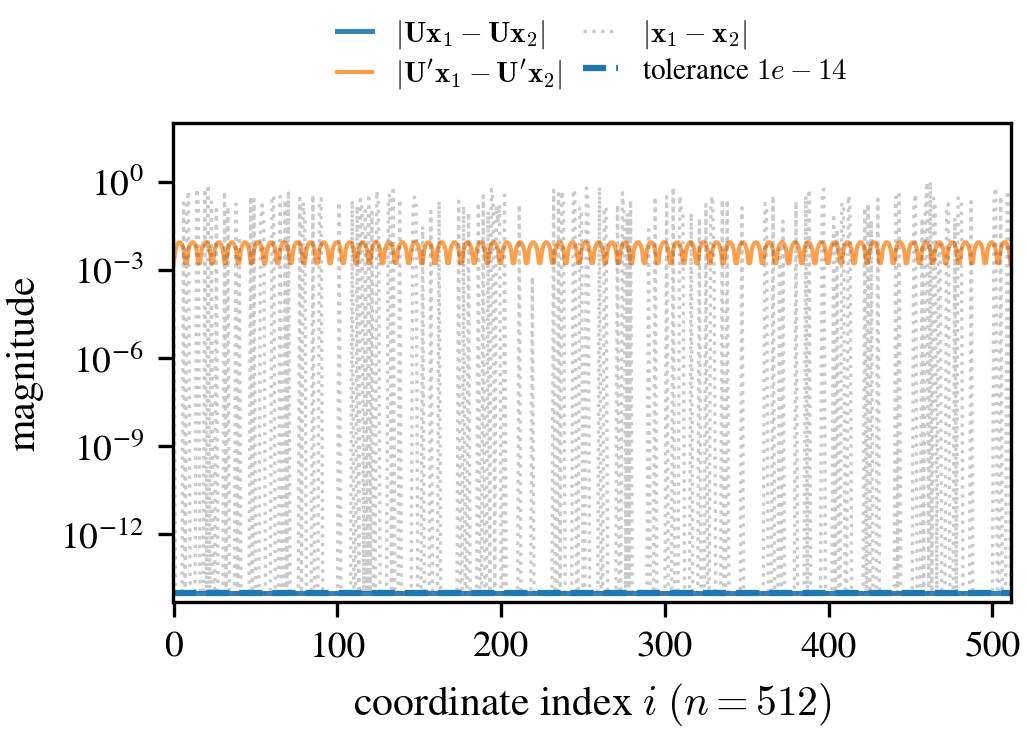}
  \label{fig:dct-sub-b}
\end{minipage}
}

\par\vspace{0.25em}

\subfloat[]{%
\begin{minipage}[t]{0.98\textwidth}
  \centering
  \includegraphics[width=\linewidth,height=0.21\textheight,keepaspectratio]{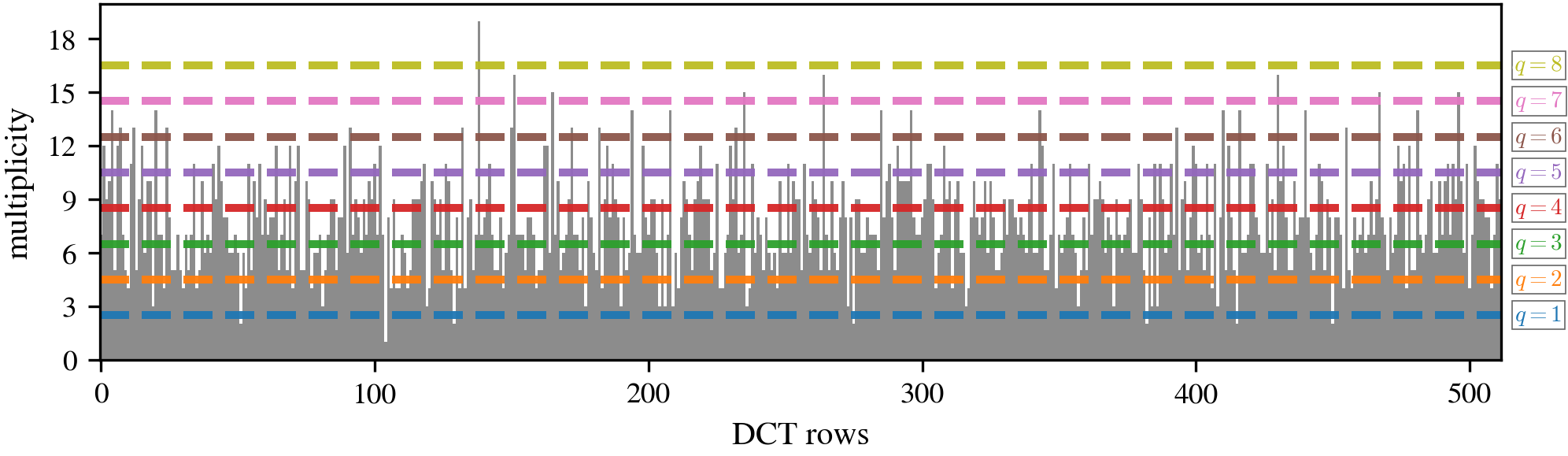}
  \label{fig:dct-sub-c}
\end{minipage}
}

\caption{Oversampled DCT simulations. 
(a) Multiplicities of sampled DCT rows for $\bA\in\R^{140\times512}$ with $q=1$. 
(b) Comparison of the correct $2q$ multiplicity threshold with an incorrect $q$ threshold. 
(c) Number of robust DCT atoms for $\bA\in\R^{3600\times512}$ as $q$ varies.}
\label{fig:dct-2plus1}
\end{figure}
Consider a stylized transform-measurement model, motivated by DCT-based estimation problems~\cite{cruz2016dct}. 
Let $\bA\in\R^{m\times n}$ be formed by sampling rows from an $n\times n$ orthonormal DCT matrix independently and uniformly with replacement. 
Given $\by=\bA\xS+\be$ with $\|\be\|_0\le q$, we ask which DCT components of $\xS$ remain robust.

In this setting, the robust subspace has a simple counting characterization: a DCT atom belongs to $\cR$ iff its row appears in $\bA$ with multiplicity strictly larger than $2q$. 
Atoms with multiplicity at most $2q$ can be removed by a $2q$-row deletion, while atoms with multiplicity greater than $2q$ survive every such deletion. 
Thus $\cR$ is computable in $O(m)$ time by counting row multiplicities.

\Cref{fig:dct-2plus1} illustrates this characterization through two simulations. 
First, we sample $140$ rows from the $512\times512$ DCT matrix, giving $\bA\in\R^{140\times512}$. 
For $q=1$, \Cref{fig:dct-sub-a} shows that only two DCT atoms remain robust, and no atoms remain robust for $q>1$. 
\Cref{fig:dct-sub-b} illustrates the necessity of the $2q$ threshold: there exist distinct signals $\bx_1,\bx_2$ and $q$-sparse corruptions $\be_1,\be_2$ such that
\[
    \bA\bx_1+\be_1=\bA\bx_2+\be_2,
\]
with $\bU\bx_1=\bU\bx_2$ for the correct $2q$-threshold projector, but $\bU'\bx_1\neq\bU'\bx_2$ for the incorrect $q$-threshold projector.

Second, we sample $3600$ rows from the same DCT matrix, giving $\bA\in\R^{3600\times512}$. 
\Cref{fig:dct-sub-c} shows the number of robust DCT atoms as $q$ varies. 
Even in this oversampled setting, uniform sampling is fragile against adversarial corruption: only one atom remains robust at $q=8$, and the robust subspace collapses to zero at $q=9$.

\section{Relation to other recovery problems} \label{sec:connections}
\subsection{Distinction from Robust Subspace Recovery}
The phrase "robust subspace'' suggests a possible connection with robust subspace recovery (RSR)~\cite{hardt2013algorithms}, but the two problems have different objectives. In RSR, one is given a corrupted data matrix and seeks an unknown inlier subspace containing most of its rows. In this work, the measurement matrix $\bA$ is known and uncorrupted; the corruption occurs in the measurement vector $\by=\bA\xS+\be$. Our goal is therefore not to recover an inlier subspace of $\bA$, but to identify the subspace of signal information that is invariant under every $q$-sparse measurement corruption.

A tempting but incorrect analogy is to apply RSR directly to the rows of $\bA$: with outlier fraction $\alpha=q/m$, one might seek the smallest subspace containing at least $m-q$ rows of $\bA$. The following result shows that this RSR-type object does not coincide with our robust subspace $\cR$.

\begin{theorem}\label{thm:rsc-vs-rsr-counterexample}
There exist integers $m,n$, matrix $\bA \in \R^{m\times n}$, and integer $q < m/2$ such that the subspace $\cR = \bigcap_{\substack{T\subseteq[m]\\ |T|=m-2q}} \rspan{\bA_T}$ does not coincide with any solution of the following optimization problem:
\begin{equation}\label{eq:RSR-type}
    \min_{\cV\subseteq\R^n}\ \dim(\cV)
\text{ s.t. }
\bigl|\{i\in[m]:\ba_i\in\cV\}\bigr| \ge m-q,
\end{equation}
where $\ba_i^\tr$ are the rows of $\bA.$
\end{theorem}

\begin{proof}
Let $q=1$, and consider the matrix
\(
\bA= 
\begin{bmatrix}
1 & 2 & 0 & 0 & 0 \\
0 & 0 & 1 & 2 & 0 \\
0 & 0 & 0 & 0 & 1
\end{bmatrix}^{\!\top}.
\)
Here,
$\cR = \{0\}$ 
while
$\cV_1 := \Span\{(1,0,0),(0,1,0)\}$ is a feasible solution of \eqref{eq:RSR-type}. Indeed, \(\dim(\cV_1)=2\) and
the first four rows of $\bA$ lie in $\cV_1$, so
\(|\{i : \ba_i \in \cV_1\}|=4 \geq 5-1\).
Moreover, no $1$-dimensional subspace can be feasible since no set of $4$ rows of $\bA$ are colinear.
Any $\cV$ with $\dim(\cV)=1$ satisfies
\(|\{i : a_i \in \cV\}| \le 2 < 4\) and is infeasible.
Now, since $\cR \neq \cV_1$, this provides the desired counterexample.
\end{proof}

\subsection{Assumption-free compressed sensing}
Here we discuss how our solution set viewpoint can be used to derive similar assumption-free tight recoverability results in compressed sensing (or sparse recovery). This illustrates the utility of the research methodology developed in this work to obtain results in different recovery problems. Just as in sparse error correction, compressed sensing has also been dominated by the exact-recovery framework, which seeks sufficient conditions for the exact recovery of a sparse signal from its linear measurements \cite{candes2006near, donoho2006compressed}.

A well known necessary and sufficient condition (\cite[Corollary 1]{donoho2003optimally}, which we state as \Cref{thm:spark} below) for uniqueness of a sparse solution $\eS$ of a linear system $\bz = \bF \eS$ uses the notion of the \emph{spark} of $\bF$. It is defined as \(
\spark(\bF) := \min_{\bv \in \ker(\bF)\setminus \{0\}} \|\bv\|_0,\) and it gives the smallest number of linearly dependent columns of the matrix $\bF \in \R^{l\times m}$ (see \cite[Definition 1]{donoho2003optimally}).
    
\begin{theorem}[Corollary 1 in  \cite{donoho2003optimally}]\label{thm:spark}
    For $\bF \in \R^{l \times m}$ and integer $q\geq0,$ any $q$-sparse $\be \in \R^m$ (i.e., $\|\be\|_0 \leq q)$ is the unique $q$-sparse representation of $\bz = \bF \be$ if and only if $q<\spark(\bF)/2.$
\end{theorem}

In our result \Cref{thm:sparse-result}, we generalize \Cref{thm:spark} and delineate the best sparse recovery result for any arbitrary matrix $\bF.$ As in \Cref{sec:main}, here optimality is in terms of inclusion-wise minimal solution sets.
We prove that the maximum information about any $q$-sparse vector $\eS \in \R^m$ recoverable from linear measurements $\bF\eS$ is precisely the set $(\eS + \ker(\bF))  \cap \Sigma_q,$ where $\Sigma_q=\{\bv\in\R^m:\|\bv\|_0\le q\}$ denotes the set of $q$-sparse vectors in $\R^m$. 
 
The steps are the following, as before: (i) In \Cref{def:consistency}, we first characterize functions that remain consistent under consistency of measurements generated by $q$-sparse signals; then, (ii) in \Cref{def:cons-sol-set}, we characterize the sparsity-obeying solution sets generated by such consistent functions; and finally, (iii) in \Cref{thm:sparse-result}, we identify the smallest solution set---and therefore the set giving the maximum information---containing the $q$-sparse signal $\eS$ that can be identified from linear measurements $\bz = \bF\eS.$

\begin{definition}\label{def:consistency}
Let $\Sigma_q=\{\bv\in\R^m:\|\bv\|_0\le q\}$ denote the set of $q$-sparse vectors in $\R^m$. For $\bF\in\R^{l\times m}$ and integer $q<m/2$, a function $\sE:\R^m\to\cZ$ (arbitrary codomain) is said to be \emph{$(\bF,q)$-consistent} if for every $\be_1,\be_2\in\Sigma_q$,
\begin{equation}\label{eq:consistency-criterion}
    \bF\be_1=\bF\be_2 \implies \sE(\be_1)=\sE(\be_2).
\end{equation}
\end{definition}
Note that in \Cref{def:consistency}, the functions $\sE$ need to obey the consistency criterion \eqref{eq:consistency-criterion} over only the set of sparse vectors $\Sigma_q.$ They may take arbitrary values elsewhere.

\begin{definition}\label{def:cons-sol-set}
Let $\bF\in\R^{l\times m}$ and integer $q<m/2$. Also, let $\Sigma_q=\{\bv\in\R^m:\|\bv\|_0\le q\}$. Then, for any $(\bF,q)$-consistent function $\sE$ and any $\be^\star\in\Sigma_q$, the \emph{$(\bF,q)$-consistent solution set containing $\eS$} is
\begin{equation}
E(\sE, \eS):=\{\be\in\Sigma_q:\sE(\be)=\sE(\eS)\}.    
\end{equation}
\end{definition}

In the following theorem, we prove that there exists no $(\bF, q)$-consistent function $\sE$ that is more informative than $\bF$. That is, $\bF$ generates the smallest $(\bF,q)$-consistent solution set. 

\begin{theorem}\label{thm:sparse-result}
Consider a matrix $\bF \in \R^{l \times m}$ and integer $q<m/2.$ Also, let $\Sigma_q=\{\bv\in\R^m:\|\bv\|_0\le q\}$. Then, the $\bv \to \bF\bv$ map is $(\bF,q)$-consistent by definition. Moreover, for any $(\bF,q)$-consistent $\sE:\R^m\to\cZ$ and $\be^\star\in\Sigma_q$, 
\begin{equation}
    \{\be\in\Sigma_q:\sE(\be)=\sE(\eS)\} \supseteq  (\be^\star + \ker(\bF)) \cap \Sigma_q,
\end{equation}
with equality for $\sE(\bv)=\bF\bv.$
\end{theorem}
\begin{proof}
Firstly, note that \(E(\bF, \be^\star) = \{\be \in \Sigma_q : \bF \be = \bF \be^\star\} = (\be^\star + \ker(\bF)) \cap \Sigma_q\).
We need to show that for any ($\bF,q$)-consistent \(\sE\), 
\begin{equation}
E(\sE, \be^\star) \supseteq E(\bF, \be^\star).    
\end{equation}
Indeed, this is immediate from \Cref{def:consistency}, since for any
$\be \in E(\bF,\be^\star)$, we have $\bF\be = \bF\be^\star$, and thus,
by consistency of $\sE$, we have $\sE(\be) = \sE(\be^\star)$. 
Hence $\be \in E(\sE,\be^\star)$.
\end{proof}
With \Cref{thm:sparse-result}, we know what is the information theoretic limit of sparse signal recovery when the measurements are obtained using any arbitrary signal encoding matrix $\bF.$ For any sparse signal $\eS,$ the set $(\be^\star + \ker(\bF)) \cap \Sigma_q$ arises naturally as an inverse problem solution set. \Cref{thm:sparse-result} establishes that no smaller set can be recovered. Indeed, when considering the case of full recovery of $\eS$, \Cref{cor:sparse-spark} below shows that whenever the spark criterion in \Cref{thm:spark} (Corollary 1 in  \cite{donoho2003optimally}) holds, \Cref{thm:sparse-result} gives the smallest solution set as $\{\eS\}$, thereby verifying its tightness.

\begin{corollary}\label{cor:sparse-spark}
Let \(\bF \in \R^{l \times m}\) and integer \(q < m/2\).
Then,
\begin{equation}
q < \spark(\bF)/2 
\;\Leftrightarrow\; (\eS + \ker(\bF)) \cap \Sigma_q = \{\eS\}
\quad \forall \eS \in \Sigma_q,
\end{equation}
where \(\spark(\bF) = \min_{\bv \in \ker(\bF)\setminus \{0\}} \|\bv\|_0\).
\end{corollary}

\begin{proof}
($\Rightarrow$) Assume \(q < \spark(\bF)/2\). Let \(\eS \in \Sigma_q\) and suppose \(\be \in (\eS + \ker(\bF)) \cap \Sigma_q\) with \(\be \neq \eS\). Then \(\bv := \be - \eS \in \ker(\bF) \setminus \{0\}\), so \(\|\bv\|_0 \ge \spark(\bF)\). But \(\|\bv\|_0 \le \|\be\|_0 + \|\eS\|_0 \le 2q\), implying \(\spark(\bF) \le 2q\), a contradiction. Thus, \((\eS + \ker(\bF)) \cap \Sigma_q = \{\eS\}\).

($\Leftarrow$) Assume \((\eS + \ker(\bF)) \cap \Sigma_q = \{\eS\}\) for all \(\eS \in \Sigma_q\). Suppose, for contradiction, \(q \ge \spark(\bF)/2\). Then, there exists a \(\bv \in \ker(\bF) \setminus \{0\}\) with \(s := \|\bv\|_0 \le 2q\). Let \(S := \supp(\bv)\). Choose \(T \subseteq S\) with \(|T| = t := \lceil s/2 \rceil \le q\), so \(|S \setminus T| = \lfloor s/2 \rfloor \le q\). We will now construct a pair $\be,\eS\in \Sigma_q$ supported on $S \setminus T $ and $T$ respectively, such that $\be \in (\eS + \ker(\bF)) \cap \Sigma_q$, while clearly $\be \neq \eS.$

Define \(\eS \in \R^m\) by \(e^\star_j = -v_j\) if \(j \in T\), else 0. Then \(\|\eS\|_0 = t \le q\), so \(\eS \in \Sigma_q\). Now, define \(\be := \eS + \bv\). Clearly, \(\supp(\be) = S \setminus T\), so \(\be  \in \Sigma_q\), while \(\be \neq \eS\). Finally, since $\bv \in \ker(\bF),$ \(\be \in (\eS + \ker(\bF)) \cap \Sigma_q\), a contradiction. Thus, \(q < \spark(\bF)/2\).
\end{proof}

\section{Conclusion}
In this paper, we go beyond the conventional exact-recovery framework, and determine what inherent robust information about an input vector $\xS$ can be identified from $\by = \bA \xS + \be,$ where $\be$ with $\|\be\|_0\leq q$ may be any sparse adversarial corruption. Specifically, we show that the robust information is given by the affine subspace $\xS + \ker(\bU),$ where $\bU$ is the orthogonal projection onto the subspace common to the rowspaces of all $2q$-row deleted submatrices of $\bA.$ This subspace is referred to as the robust subspace, and it determines whether for a given $(\bA,q)$ we could have exact, partial, or trivial robust recovery. We further study the computability of the robust subspace and prove that computing it is NP-Hard for arbitrary matrices, motivating the study of structured measurement models. We also show that i.i.d. Gaussian matrices admit an $O(1)$ computation of the robust subspace due to a phase transition phenomenon, while matrices formed by oversampling from orthonormal systems admit an $O(m)$ computation scheme. Along the way, we developed a principled framework to reason concretely about the robustness inherent in the measurement model based on the notions of robust functions and the ambiguity set. These concepts suggest several directions for future work, including  nonlinear models, $\ell_1$-recoverability, approximate recoverability in the presence of noise, matrix completion, and other settings. 

\section*{Declaration of prior dissemination}
A preliminary version of part of this work was accepted for presentation at the European Signal Processing Conference (EUSIPCO) 2026, with a preprint available as arXiv:2510.24215v3. The present preprint corresponds to arXiv:2510.24215v5. The conference version included the robust subspace characterization, \(\ell_0\)-decoding recovery, an algorithm for the robust orthogonal projection, and stylized applications. The present manuscript substantially extends it with the NP-hardness result (\Cref{sec:hardness}), Gaussian random matrix phase transition (\Cref{sec:gauusian-random-matrices}), and connections to robust subspace recovery and compressed sensing (\Cref{sec:connections}). The introduction, motivation, and overall presentation have also been significantly revised and expanded.

\section*{Funding}
This work was supported by the Indo-French Centre for the Promotion
of Advanced Research (CEFIPRA) [Project 7102-1].

\section*{Declaration of competing interest}
The authors declare no known competing financial interests or personal relationships that could have influenced the work reported in this paper.

\section*{CRediT authorship contribution statement}
Vishal Halder is the lead author and contributed: formal analysis, investigation, methodology, validation, visualization, and writing -- original draft.
Alexandre Reiffers-Masson, Abdeldjalil A\"issa-El-Bey, and Gugan Thoppe were involved with: conceptualization, supervision, and writing -- review and editing.

\section*{Data availability}
No external datasets were used in this study. 
The numerical illustrations use synthetically generated data following the procedures described in the manuscript.

\section*{Declaration of generative AI use}
The authors used ChatGPT (OpenAI) for figure creation, subsequently reviewed and edited them, and take full responsibility for this manuscript.

\section*{Acknowledgements}
The authors would like to thank the Indian Institute of Science for hosting
some of them for brief visits.

\bibliographystyle{plain}
\bibliography{refs}

\end{document}